\documentclass[aps,prd,superscriptaddress,12pt,tightenlines,nofootinbib]{revtex4}
\usepackage{amsmath,amssymb,amsthm}
\usepackage{graphicx}
\usepackage{dcolumn}

\usepackage{color}

\newtheorem{Theorem}{Theorem}

\newtheorem{Lemma}{Lemma}

\def\refa#1{(\ref{#1})}

\def\ra{\rightarrow}

\def\C{\mathcal{C}}
\def\N{\mathbb{N}}
\def\R{\mathbb{R}}
\def\RTX{\R_+^{1+3}}
\def\Linf{{L^\infty}}
\def\Linfx#1{{L^\infty_{#1}}}
\def\Linftx#1{{L^\infty_{1,#1}}}
\def\Linftau#1#2{{L^\infty_{#1,#2}}}
\def\LinfX{{L^\infty(\R_+)}}
\def\LinfTX{{L^\infty(\R_+\times\R_+)}}

\def\<{\langle}
\def\>{\rangle}
\def\norm#1{\<#1\>}
\def\nnorm#1{(1+|#1|)} 
\def\w#1{\widetilde{#1}}

\def\d{\partial}

\def\eps{\varepsilon}
\def\la{\lambda}
\def\Ocal#1{\mathcal{O}\left(#1\right)}

\begin{document}

\title{Linear and nonlinear tails II: \\ exact decay rates in spherical symmetry}

\author{Nikodem Szpak}
\affiliation{Max-Planck-Institut f{\"u}r Gravitationsphysik, Albert-Einstein-Institut, Potsdam,
Germany}
\author{Piotr Bizo\'n}
\affiliation{M. Smoluchowski Institute of Physics, Jagiellonian University, Krak\'ow, Poland}
\author{Tadeusz Chmaj}
\affiliation{H. Niewodniczanski Institute of Nuclear
   Physics, Polish Academy of Sciences,  Krak\'ow, Poland}
   \affiliation{Cracow University of Technology, Krak\'ow,
    Poland}
\author{Andrzej Rostworowski}
\affiliation{M. Smoluchowski Institute of Physics, Jagiellonian University, Krak\'ow, Poland}

\date{\today}

\begin{abstract}
We derive  the exact late-time  asymptotics for small spherically symmetric solutions of
nonlinear wave equations with a potential.
 The dominant tail is shown to result from the competition between  linear and nonlinear effects.
\end{abstract}


\maketitle

\section{Introduction}

We consider linear and nonlinear wave equations with a potential term
\begin{equation} \label{wave-eq}
   \Box u + \la V u = F(u)\,, \qquad \Box=\partial_t^2-\Delta\,,
\end{equation}
in three spatial dimensions for spherically symmetric  initial data
\begin{equation} \label{init-data}
  u(0,r)=f(r),\qquad \d_t u(0,r)=g(r),\qquad r:=|x|
\end{equation}
with $f,g$ of compact support. The spherical symmetry of the initial data is preserved in
evolution so $u=u(t,r)$. We are interested in the asymptotic behaviour of $u(t,r)$ for late times
$t\gg r$.

Our approach is based on the perturbative calculation which has been developed by the last three
authors in concrete physical applications \cite{bcr1, bcr2, bcr3} and recently put on the
rigorous ground by the first author in \cite{NS-Tails} (below referred to as part I).
In part I the convergence of the perturbation scheme was proved in a weighted space-time
$\Linf$-norm which provided pointwise estimates on the solution $u(t,r)$ in the whole spacetime.
Moreover, upper bounds on the errors (remainders of the perturbation series) for every
perturbation order were obtained. Here, we are going to combine the qualitative global
weighted-$\Linf$ estimates with the quantitative perturbation scheme in order to obtain precise
late-time asymptotics of  solutions. To this end, we first solve the linear perturbation
equations analytically up to the second (nontrivial) order (in spherical symmetry this can be
done explicitly) and show that our decay estimate is optimal. Then, we prove that the sum of all
higher-order perturbations does not modify the dominant asymptotics, hence the second order
perturbation gives the precise approximation of the tail of the solution $u$.  Along the way, we
illustrate our analytical results with numerical solutions of the initial value problem
(\ref{wave-eq}-\ref{init-data}).

The basis of our analysis is given by the theorem of Strauss and Tsutaya \cite{Strauss-T}, recently generalized by NS \cite{NS-WaveDecay}, which states that
\begin{equation}
  |u(t,x)| \leq \frac{C}{(1+t+|x|)(1+|t-|x||)^{q-1}} \qquad \forall (t,x)\in\RTX
\end{equation}
with $q:=\min(m-1,k,p-1)$ provided that the potential $V$ and the initial data $f,g$ satisfy
pointwise bounds
\begin{equation} 
  |V(x)| \leq \frac{V_0}{\nnorm{x}^k},\quad k>2\,,
\end{equation}
\begin{equation} \label{fg-bound0}
  |f(x)| \leq \frac{f_0}{\nnorm{x}^{m-1}}, \qquad
  |\nabla f(x)| \leq \frac{f_1}{\nnorm{x}^m}, \qquad
  |g(x)| \leq \frac{g_0}{\nnorm{x}^m},\qquad m>3\,,
\end{equation}
with small $V_0, f_0, f_1, g_0$ and the analytic nonlinearity satisfying for $p>1+\sqrt{2}$
\begin{equation}
  |F(u)|\leq F_1 |u|^p,\qquad |F(u)-F(v)|\leq F_2 |u-v| \max(|u|,|v|)^{p-1}\qquad \text{for } |u|, |v|<1.
\end{equation}
This is true for classical solutions \cite{Strauss-T}, i.e. for
$(f,g)\in\C^3\times\C^2$, $V\in\C^2$ and $F\in\C^2$, leading to
$u\in\C^2$ and remains true also for weak solutions, e.g. $u\in\C^0$, \cite{NS-WaveDecay}.

Here, for simplicity, we consider initial data of compact support so the decay rate $q$ is
determined solely by the spatial decay rate of the potential $k$ and the leading power of the
nonlinearity $p$. Generalization of these results to the initial data with  the fall-off
(\ref{fg-bound0}) is straightforward.

The paper is organized as follows. We first study the purely linear situation with the potential
term only. Then, we repeat the calculations for the purely nonlinear case without the potential.
Finally, we combine both results in the general case \refa{wave-eq}.

\subsection*{Notation}


We use the symbol $\<x\>:=1+|x|$ to denote the spatial weighted-$L^\infty$ norm
\begin{equation}
  \| f \|_\Linfx{m} := \| \norm{r}^m f(r) \|_\LinfX
\end{equation}
We also define a space-time weighted-$L^\infty$ norm
\begin{equation}
  \| u \|_\Linftau{s}{q} := \| \norm{t+r}^s \norm{t-r}^{q-s}  u(t,r)\|_\LinfTX\,.
\end{equation}
We will frequently use the fact that the  finiteness of  $\| u \|_\Linftau{1}{q}$ guarantees the
decay of $u$ like $1/t$ on the lightcone $t\sim r$ and like $1/t^q$ for fixed $r$ as well as
$1/r^q$ for fixed $t$. Note that functions with compact support in $\R_+$ belong to all spaces
$\Linfx{m}$ with any $m>0$, what we will denote by $\Linfx{\infty}$. Analogously
$\Linftx{\infty}$ will stand for functions that belong to $\Linftx{q}$ for any $q$.

We introduce the following notation for solutions of the wave equations. Let $I_V$ be a linear
map from the space of initial data to the space of solutions of the wave equation
\refa{wave-eq}-\refa{init-data} with $F(u)=0$, so that $u=I_V(f,g)$. For the wave equation with a
source term and zero initial data
\begin{equation}
  \Box u + Vu= F,\qquad u(0,r)=0,\qquad  \d_t u(0,r)=0,
\end{equation}
we denote the solution by $u=L_V(F)$, where $L_V$ is a linear map from the space of source
functions to the space of solutions to the above problem. Note that, due to linearity, the
solution $u$ of a wave equation with source $F$ and nonzero initial data $f,g$ is the sum of
these two contributions
\begin{equation}
  u=L_V(F)+I_V(f,g).
\end{equation}
Observe that if we put the potential term on the r.h.s. we obtain
\begin{equation}
  \Box u = -V u + F\,,
\end{equation}
which, treated as a wave equation without potential (on the l.h.s.), is formally solved by
\begin{equation}
  u=-L_0(Vu)+L_0(F)+I_0(f,g).
\end{equation}
Here the solution $u$ appears on both sides what seems to make the formula useless, but it will
allow us to formulate various iteration schemes, e.g.
\begin{equation}
  u_{n+1}=-L_0(Vu_n)+L_0(F(u_n))+I_0(f,g)
\end{equation}
for which we will study convergence in suitable $\Linftx{q}$ norms.

Finally, we define constants which arise from estimates proved in \cite{NS-WaveDecay} and
improved in \cite{NS-DecayLemma}
\begin{equation}
  C_m:= \max \left(\frac{9}{2(m-2)}, 5\right),
\end{equation}
\begin{equation}
  C_{p,q}:=2+\frac{8}{p-1}+\frac{2}{q-1}.
\end{equation}
The latter will be referred to as a bound on the allowed strength  of the potential. We wish to
emphasize that this bound, although not optimal, is not arbitrarily small but finite, which is
crucial in applications (like, for instance, the Regge-Wheeler equation describing waves
propagating on Schwarzschild geometry).

We recall some standard definitions of asymptotic calculus. The notation $f(t)=\Ocal{g(t)}$ for
$t\ra\infty$ means that there exist constants $C,T>0$ such that
\begin{equation}
  |f(t)| \leq C |g(t)|
\end{equation}
for all $t>T$. The notation $f(t)=o(h(t))$ for $t\ra\infty$ means that
\begin{equation}
  \lim_{t\ra\infty} \frac{f(t)}{h(t)} = 0.
\end{equation}
We will also use the symbol $f(t)\cong g(t)$ for an asymptotic approximation, as a shorthand to
$f(t)=g(t)[1+o(1)]$ as $t\ra\infty$.
In case when we write $f(t)\cong c t^{-q}$ and the constant $c$ may become zero, this notation should be read as $f(t)= c t^{-q} + o(t^{-q})$.

\section{Linear case with potential}
First, we consider the linear wave equation
\begin{equation} \label{V:wave-eq}
   \Box u + \la V(r) u = 0
\end{equation}
with initial data \refa{init-data}, where $f(r)$ and $g(r)$ are supported on the interval
$r\in[0,R]$. We assume that $V(r)\cong V_0/r^k$ for $r\gg 1$ and $\la>0$ is a small parameter,
bounded by some finite constant $C_V>0$ (which will be defined later). Moreover, we assume that
the potential $V$ and the initial data $f,\nabla f,g$ are (at least) continuous and satisfy
\begin{equation} \label{V-bound}
  \|V\|_\Linfx{k} = 1
\end{equation}
and
\begin{equation} \label{fg-bound}
  f_0:=\|f\|_\Linfx{k}, \qquad
  f_1:=\|\nabla f\|_\Linfx{k+1}, \qquad
  g_0:=\|g\|_\Linfx{k+1}
\end{equation}
with $f_0, f_1, g_0<\infty$ for some $k>2$.

\subsection{Perturbation series}

We define the perturbation series
\begin{equation}\label{l_sum}
   u = \sum_{n=0}^\infty \la^n v_n\,.
\end{equation}
Inserting (\ref{l_sum}) into equation \refa{V:wave-eq} we get the following perturbation scheme
\begin{alignat}{4} \label{pert-V-0}
  \Box v_0 &= 0,&\qquad (v_0,\dot{v}_0)(0)&=(f,g)&\qquad&\ra&\qquad v_0&=I_0(f,g) \\ \label{pert-V-n}
  \Box v_{n+1} &= -V v_n,&\qquad (v_{n+1},\dot{v}_{n+1})(0)&=(0,0)&\qquad&\ra&\qquad v_{n+1} &= -L_0(V
  v_n)\,.
\end{alignat}
Due to linearity of \refa{V:wave-eq} it turns out that the partial sums
\begin{equation}
   u_n:=\sum_{k=0}^n \la^k v_k, \qquad n\geq 0
\end{equation}
satisfy the following iteration scheme
\begin{equation}
   u_{-1} := 0\\
\end{equation}
\begin{equation}
   u_n:= I_0(f,g) - \la L_0(V u_{n-1}),\qquad n\geq 0.
\end{equation}
Then, from part I, we have the following
\begin{Theorem} \label{Th:V}
For $f,g$ and $V$ as above and any $k>2$, the sequence $u_n$ converges (in norm) in $\Linftx{k}$
provided that $\la<C_{k,k}^{-1}$. The limit $u:=\lim_{n\ra\infty} u_n$ satisfies
\begin{equation}
  |u(t,r)| \leq \frac{C}{\norm{t+r}\norm{t-r}^{k-1}},\qquad \forall (t,r)
\end{equation}
where a positive constant $C$ depends only on $f_0, f_1, g_0, \la$ and $k$.
\end{Theorem}

From the proof of theorem \ref{Th:V} it follows that
\begin{equation}
  \| v_n\|_\Linftx{k} = \frac{\|u_{n}-u_{n-1}\|_\Linftx{k}}{\la^n} \leq (C_{k,k})^{n} \|I_0(f,g)\|_\Linftx{k},
\end{equation}
hence $v_n \in \Linftx{k}$ for all $n\geq 0$. At the lowest order $v_0=u_0$ we have an
arbitrarily fast decay estimate,  $v_0\in\Linftx{\infty}$,  as follows from Huygens' principle.
All higher-order terms $v_n (n=1,2,...)$ contain  contributions from the backscattering off the
potential and are only in $\Linftx{k}$. Since $u\in\Linftx{k}$, we expect that all $u_n$ starting
from $u_1\in\Linftx{k}$ predict qualitatively correct asymptotic behaviour of $u$.

\subsection{Optimal decay estimate}

\begin{Theorem} \label{Th:lin-v1}
Under the above assumptions, for $t\gg r+R$, we have
\begin{equation}
  v_1(t,r) \cong c_1 t^{-k}\,,
\end{equation}
where the constant $c_1$ is given by (\ref{eq:c1}).
\end{Theorem}
\begin{proof}
For $v_0$ satisfying \refa{pert-V-0} from lemma \ref{lemma:sol} we have
\begin{equation} \label{eq:v0}
u_0(t,r) = v_0(t,r) = \frac{h(t-r)-h(t+r)}{r}\,,
\end{equation}
where $h$ is given by \refa{lemma:sol:h}. To solve equation \refa{pert-V-n} we use the Duhamel
representation for the solution of the inhomogeneous equation $\Box v = N(t,r)$ with zero initial
data
\begin{equation} \label{a2}
v(t,r) = \frac{1}{2 r} \int\limits_{0}^{t} d\tau \int\limits_{|t-r-\tau|}^{t+r-\tau} \rho
N(\tau,\rho) d \rho.
\end{equation}
This formula can be easily obtained by integrating out the angular variables in the standard
formula $\phi=G^{ret} * N$ where $G^{ret}(t,x) = (2 \pi)^{-1} \Theta(t) \delta (t^2-|x|^2)$ is
the retarded Green's function of the wave operator in $3+1$ dimensions. It is convenient to
express (\ref{a2}) in terms of null coordinates $\xi = \tau + \rho$ and $\eta = \tau - \rho$
\begin{equation} \label{duh1}
v(t,r) = \frac{1}{4 r} \int\limits_{|t-r|}^{t+r} d\xi \, \int\limits_{-\xi}^{t-r} d \eta \,
\frac{(\xi-\eta)}{2} \w{N}(\xi, \eta),
\end{equation}
where $\w{N}(\xi,\eta):=N\left(\frac{\xi+\eta}{2},\frac{\xi-\eta}{2}\right)=N(\tau,\rho)$. Using
this representation we get from \refa{pert-V-n}
\begin{equation} \label{eq:v1(1)}
v_1(t,r) = - \frac {1} {4 r} \int\limits_{|t-r|}^{t+r} d\xi \, \int \limits_{-\xi}^{t-r} d \eta
\, \frac{(\xi-\eta)}{2} V(\rho(\xi,\eta)) \w{v}_0(\xi,\eta).
\end{equation}
Since the initial data $f,g$ are supported on $[0,R]$, the function $h(x)$ is supported on
$[-R,R]$. Then, for $t>r+R$ equation (\ref{eq:v1(1)}) simplifies to
\begin{equation}
  v_1(t,r) = - \frac {1} {4 r} \int \limits_{-R}^{+R} d \eta \, h(\eta)
  \int\limits_{t-r}^{t+r} d\xi \, V(\rho(\xi,\eta)).
\end{equation}
Next, we write the potential in the form $V(r)=r^{-k}[V_0+w(r)]$ with $w(r)\ra 0$ as $r\ra\infty$
at any rate (i.e. $w(r)=o(1)$ for $r\gg 1$). Then
\begin{equation} \label{eq:v1(2)}
  v_1(t,r) = - \frac {1} {4 r} \int \limits_{-R}^{+R} d \eta \, h(\eta)
  \int\limits_{t-r}^{t+r} d\xi \, \frac{2^k}{(\xi-\eta)^{k}} [V_0+w(\xi-\eta)]
  \equiv \widehat{v}_1(t,r) + \w{v}_1(t,r)\,,
\end{equation}
where
\begin{align}
  \widehat{v}_1(t,r) &= - \frac{2^{k-2}}{r} V_0 \int \limits_{-R}^{+R} d \eta \, h(\eta)
  \int\limits_{t-r}^{t+r} d\xi \, (\xi-\eta)^{-k} \\
  \w{v}_1(t,r) &= - \frac{2^{k-2}}{r} \int \limits_{-R}^{+R} d \eta \, h(\eta)
  \int\limits_{t-r}^{t+r} d\xi \, (\xi-\eta)^{-k} w(\xi-\eta).
\end{align}
Using lemma \ref{lemma:integral} for $t\gg r+R$ and $k>2$ we get
\begin{equation} \label{eq:v1(3)}
\widehat{v}_1(t,r) = \frac{c_1}{t^k} + \Ocal{\frac{r+R}{t^{k+1}}},
\end{equation}
with
\begin{equation} \label{eq:c1}
c_1 = - 2^{k-1} V_0 \int \limits_{-R}^{+R} h(\eta)\,d\eta.
\end{equation}
Using lemma~2 again, we get an estimate for $\w{v}_1$
\begin{equation}\label{est_v1_tilde}
\begin{split}
  |\w{v}_1(t,r)| &\leq  \frac{2^{k-2}}{r} \int \limits_{-R}^{+R} d \eta \, |h(\eta)|
  \sup_{t-r\leq \zeta\leq t+r} |w(\zeta-\eta)| \int\limits_{t-r}^{t+r} d\xi \, (\xi-\eta)^{-k}  \\
  &\leq  \sup_{t-r-R\leq \zeta\leq t+r+R} |w(\zeta)|\; \frac{2^{k-2}}{r}
  \int \limits_{-R}^{+R} d \eta \, |h(\eta)| \int\limits_{t-r}^{t+r} d\xi \, (\xi-\eta)^{-k}  \\
  &= \sup_{t-r-R\leq \zeta\leq t+r+R} |w(\zeta)|\; \left[
    \frac{\w{c_1}}{t^k} + \Ocal{\frac{r+R}{t^{k+1}}}\right],
\end{split}
\end{equation}
where
\begin{equation}
\w{c}_1 = - 2^{k-1} \int \limits_{-R}^{+R} |h(\eta)|\,d\eta.
\end{equation}
Note that the prefactor in (\ref{est_v1_tilde}) vanishes asymptotically for large $t$
\begin{equation}
  \lim_{t\ra\infty} \sup_{t-r-R\leq \zeta\leq t+r+R} |w(\zeta)| = 0,
\end{equation}
hence, for $t\gg r+R$ we have
\begin{equation}
  v_1(t,r) = \frac{c_1}{t^k} + \Ocal{\frac{r+R}{t^{k+1}}} + o(1)\cdot\left[
    \frac{\w{c}_1}{t^k} + \Ocal{\frac{r+R}{t^{k+1}}}\right]
  = \frac{c_1}{t^k} + o\left(\frac{1}{t^k}\right).
\end{equation}
If the potential behaves like $V(r) = V_0 r^{-k} + W(r)$ with $W(r) = \Ocal{r^{-k-1}}$ for $r\gg 1$,
it follows from (\ref{eq:v1(3)}) that
\begin{equation}
\begin{split}
  v_1(t,r) &= \frac{c_1}{t^k} + \Ocal{\frac{r+R}{t^{k+1}}}
  + \Ocal{\frac{\w{c}_1}{t^{k+1}}} + \Ocal{\frac{r+R}{t^{k+2}}}\\
  &= \frac{c_1}{t^k} + \Ocal{\frac{1+r+R}{t^{k+1}}},
\end{split}
\end{equation}
which gives the more detailed information about the sub-leading term.
\end{proof}

\begin{Theorem} \label{Th:lin-u}
Under the assumptions of theorem \ref{Th:lin-v1}, for $t\gg r+R$, we have
\begin{equation}
  u(t,r)\cong \la v_1(t,r) [1+\Ocal{\la}]\,,
\end{equation}
hence
\begin{equation}
  u(t,r) \cong C t^{-k},\qquad C=\la c_1 + \Ocal{\la^2}\,.
\end{equation}
\end{Theorem}
\begin{proof}
Knowing that the perturbation series converges for some $\la$ we can bound the error in the
$n$-th order relative to the exact solution by estimating the sum of all higher order terms. For
the convergent sequence $u_n$ we get from the proof of Theorem \ref{Th:V} that
\begin{equation}
  \|u-u_n\|_\Linftx{k} \leq \frac{(C_{k,k}\la)^{n+1}}{1-C_{k,k}\la} \|I_0(f,g)\|_\Linftx{k},
\end{equation}
what provides the pointwise bound on the error
\begin{equation} \label{pointwise-error-bound}
  |u(t,r)-u_n(t,r)| \leq \frac{(C_{k,k}\la)^{n+1}}{1-C_{k,k}\la}\cdot
  \frac{C_{k+1}\cdot(f_0+f_1+g_0)}{\norm{t+r}\norm{t-r}^{k-1}} \qquad \forall t,r\geq 0.
\end{equation}
For $n=1$ with $u_1=v_0+\la v_1$ we have
\begin{equation}
  |u(t,r)-v_0(t,r)-\la v_1(t,r)| \leq \frac{(C_{k,k}\la)^{2}}{1-C_{k,k}\la}\cdot
  \frac{C_{k+1}\cdot(f_0+f_1+g_0)}{\norm{t+r}\norm{t-r}^{k-1}}=:\Delta_1(t,r).
\end{equation}
A simple inequality (which follows immediately from Bernoulli's inequality)
\begin{equation} \label{ineq-Bernoulli}
  \frac{1}{(1-\zeta)^\sigma} \leq \frac{1}{1-\sigma\zeta} = 1+\frac{\sigma\zeta}{1-\sigma\zeta} \leq 2,\qquad
  \forall \zeta \leq 1/(2\sigma),\quad \sigma>\,,
\end{equation}
implies that
\begin{equation}
  \frac{1}{\norm{t-r}^q} = \frac{1}{(1+t)^q \left(1-\frac{r}{1+t}\right)^q} \leq \frac{2}{(1+t)^q}
\end{equation}
for $\zeta:=r/(1+t)\leq 1/(2q)$, hence it holds for all $t\geq 2q r$. The error term can then be
estimated by
\begin{equation}
  \Delta_1(t,x)
  \leq 2\,(C_{k,k} \la)^2 \frac{2\, C_{k+1}\cdot(f_0+f_1+g_0)}{(1+t)^k}   \Ocal{\frac{\la^2}{t^k}},
\end{equation}
where we have  used twice the inequality \refa{ineq-Bernoulli} for $t\geq 2(k-1) r$ and $\la \leq
1/(2\,C_{k,k})$.

From  Huygens' principle for \refa{pert-V-0} with initial data of compact support it follows that
$v_0(t,r)=0$ for $t>r+R$, hence for every $r\geq 0$ and sufficiently large $t>\max[r+R,2(k-1)r]$
we have
\begin{equation}
\begin{split}
  \left|u(t,r)-\la v_1(t,r)\right| &\leq
  \left|u(t,r)-v_0(t,r)-\la v_1(t,r)\right| + |v_0(t,r)|
  =  \Ocal{\frac{\la^2}{t^k}}\,,
\end{split}
\end{equation}
and
\begin{equation}
\begin{split}
  \left|u(t,r)-\la\frac{c_1}{t^k}\right| &\leq
  \left|u(t,r)-\la v_1(t,r)\right| + \la \left|v_1(t,r)-\frac{c_1}{t^k}\right|\\
  &=  \Ocal{\frac{\la^2}{t^k}} + o\left(\frac{\la}{t^{k}}\right),
\end{split}
\end{equation}
where we have used the result of theorem \ref{Th:lin-v1}, eq. \refa{eq:v1(3)}. Therefore
\begin{equation} \label{prediction}
  u(t,r)\cong \frac{C}{t^k},\qquad C=\la c_1 + \Ocal{\la^2}.
\end{equation}
\end{proof}
This gives the precise quantitative information about the late-time tail of $u(t,r)$ and shows
that the estimate in theorem \ref{Th:V} is optimal (for $t\gg r$) (see Table~1 and Fig.~1 for the
numerical verification).

\begin{table}[h]
\centering
\begin{tabular}{|c c||c|c|c|c||c|}
\hline
 & &\multicolumn{2}{c|}{$\lambda V_0=10^{-1}$}&\multicolumn{2}{c|}{$\lambda V_0=10^{-3}$} & $V_0$ \\
\cline{3-6}
 & & Theory & Numerics & Theory & Numerics & \\
\hline \hline
{$k=3$} & Exponent & 3.0 & 2.9996 & 3.0 & 3.0000 & 0.3485 \\
       & Amplitude & $-3.5449 \times 10^{-1}$ & $-3.0429 \times 10^{-1}$ & $-3.5449
       \times 10^{-3}$ & $-3.5394 \times 10^{-3}$ & \\
\hline
{$k=4$} & Exponent & 4.0 & 4.00001 & 4.0 & 4.00000 & 0.2339 \\
       & Amplitude & $-7.0898 \times 10^{-1}$ & $-6.6885 \times 10^{-1}$ & $-7.0898
        \times 10^{-3}$ & $-7.0856 \times 10^{-3}$ & \\
\hline
{$k=5$} & Exponent & 5.0 & 5.00000 & 5.0 & 5.00000 & 0.1560 \\
        & Amplitude & $-1.4179 $ & $-1.3745$ & $-1.4179 \times 10^{-2}$ & $-1.4175 \times 10^{-2}$ & \\
\hline
\end{tabular}
\caption{Linear case with a potential term $\lambda V(r) = \lambda V_0 \frac {\tanh^{k+2}(r)}
{r^k}$.
 The values of  $V_0$ follow from the normalization condition (\ref{V-bound}).
   The results are obtained for the initial data
    of the form:
$f(r)=0$, $g(r)=4(r^2-1) \exp\left(-r^2\right)$, which corresponds to $h(z)=z^2
\exp\left(-z^2\right)$ (see \ref{lemma:sol:h}). The number at the  Theory-Amplitude entry gives
the value of $\lambda c_1$ with $c_1$ defined in (\ref{eq:c1}). Note that for $\lambda\,V_0 =
10^{-1}$ the values of $\lambda$ are actually greater than the
   convergence radius of perturbation series  obtained
  in theorem \ref{Th:V}, but still, the formula (\ref{prediction}) seems to work very well.
} \label{tab:pot}
\end{table}

\begin{figure}[h]
\centering
\includegraphics[width=0.7\textwidth]{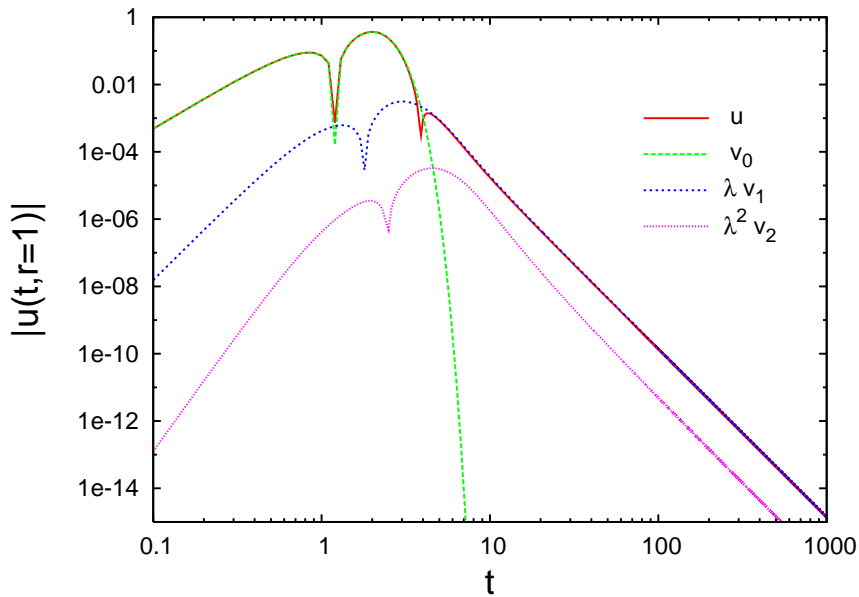}
\caption{\small{We plot (on log-log scale) the numerical solution $u(t,r=1)$ of equation
(\ref{V:wave-eq}) for $\lambda V(r)=0.1 \tanh^7{r}/r^5$ (this corresponds to $\lambda=0.64$). The
initial data are: $f(r)=0$, $g(r)=4(r^2-~1)\exp\left(-r^2\right)$, which corresponds to $h(z)=z^2
\exp\left(-z^2\right)$ (see \ref{lemma:sol:h}). The first three terms in the perturbation
expansion (\ref{l_sum}) are superimposed. In agreement with theorem~3, the tail is perfectly
approximated by $\lambda v_1$ (cf. Table~1).}}
\end{figure}

\section{Nonlinear case without a potential term}

Now, we consider the nonlinear wave equation of the form
\begin{equation} \label{Fu:wave-eq}
   \Box u = F(u)
\end{equation}
with initial data $(f,g)$ supported on the interval $r\in[0,R]$ and satisfying \refa{fg-bound}
with $f_0, f_1, g_0 < \eps$. The nonlinear term obeys $|F(u)|\leq F_1 |u|^p$ for $|u|<1$ and
$|F(u)-F(v)|\leq F_2 |u-v| \max(|u|,|v|)^{p-1}$. The second condition is satisfied e.g. for
$F(u)=u^p$ with $F_2 = p$ or for $F\in\C^1$ such that $|F'(u)|\leq F_2 |u|^{p-1}$ for $|u|<1$.

\subsection{Perturbation series}

In order  to construct a well-defined perturbation scheme to all orders we have to assume
additionally that $F(u)$ is analytic at $u=0$ and
its Taylor series starts at power $p\geq 3$
\begin{equation} \label{eq:Fu-Taylor-series}
F(u) = u^p \sum_{n=0}^\infty b_n u^n, \qquad b_0 \neq 0
\end{equation}
Then, for small initial data
\begin{equation} \label{Fu-pert-initdata}
  (u,\dot{u})(0) = (\eps f,\eps g)
\end{equation}
we introduce the perturbation series for  the solution of \refa{Fu:wave-eq}
\begin{equation} \label{Fu-pert-ser}
   u = \sum_{n=1}^\infty \eps^n v_n.
\end{equation}
Inserting this series into \refa{Fu:wave-eq} and collecting terms according to powers of $\eps$
we obtain the following perturbation scheme
\begin{alignat}{4} \label{Fu-perteq1}
  \Box v_1 &= 0,&\qquad (v_1,\dot{v}_1)(0)&=(f,g)&\quad&\ra&\quad
  v_1&=I_0(f,g) \\  \label{Fu-perteq}
  \Box v_{n+1} &= F_n(v_1,...,v_n),&\qquad (v_{n+1},\dot{v}_{n+1})(0)&=(0,0)&\quad&\ra&\quad
  v_{n+1} &= L_0(F_n(v_1,...,v_n)),
\end{alignat}
for $n\geq 1$, where $F_n$ result from collecting the nonlinear terms with the same powers of
$\eps$
\begin{equation} \label{Fu-Fn}
  F_n(v_1,...,v_n)  \sum_k 
  a^n_k v_1^{\alpha^{n,1}_k} \cdots v_n^{\alpha^{n,n}_k},
\end{equation}
where $\alpha^{n,m}_k\in\N$ satisfy $\sum_{m=1}^n m\alpha^{n,m}_k = n+1$ and $\sum_{m=1}^n
\alpha^{n,m}_k\geq p$ for every $n,k$. The coefficients $a_k^n$ are functions of $b_m$ only (see
\cite{NS-Tails} for the explicit formula).

We call this expansion the ``zero background'' case because the zero-order term $v_0$ is absent.
If a $v_0$ term was present in the series above (i.e. the summation started at $n=0$), we would
have an additional equation $\Box v_0 = F(v_0)$ which is genuinely  nonlinear (in contrast to the
above system of linear wave equations with source terms). Its solution $v_0$ represents a
``background'' around which the perturbations $v_n$ are calculated.

From part I, we have the following

\begin{Theorem} \label{Th:Fu-pert}
With $f,g$ and $F(u)$ as above for any $p\geq 3$ and sufficiently small $\eps$ the series defined
in \refa{Fu-pert-ser}-\refa{Fu-perteq} converges (in norm) in $\Linftx{p-1}$ to the solution of
equation \refa{Fu:wave-eq} with initial data \refa{Fu-pert-initdata}.
\end{Theorem}

Since the introduction of the auxiliary parameter $\eps$ in the perturbation series expansion
serves only to generate a system of linear equations equivalent to the original nonlinear
equation, we can eventually remove the parameter $\eps$ and assume that the initial data are such
that $f_0, f_1, g_0$ are sufficiently small. Then, solving the system
\refa{Fu-perteq1}-\refa{Fu-perteq} and summing up the convergent series $\sum_{n=1}^\infty v_n u$ we obtain a solution of the nonlinear wave equation \refa{Fu:wave-eq}.

\subsection{Optimal decay estimate}

The perturbation scheme \refa{Fu-perteq1}-\refa{Fu-perteq} can be written as
\begin{align}
  v_1&=I_0(f,g) \\
  v_2&=v_3=...=v_{p-1}=0\\
  v_{p} &= L_0(F_{p-1}(v_1,...,v_{p-1})) = b_0 L_0((v_1)^p) \label{Fu-perteq-p}\\
  v_{n+1} &= L_0(F_n(v_1,...,v_n)),\qquad n\geq p.
\end{align}
We have $v_1=I_0(f,g)\in\Linftx{\infty}$ and $v_n\in\Linftx{p-1}$ for $n\geq 2$.
\begin{Theorem} \label{Th:Fu-vp}
Under the above assumptions, for $t\gg r+R$, we have
\begin{equation}
  v_p(t,r) \cong d_p\, t^{-(p-1)}\,,
\end{equation}
where the constant $d_p$ is given by (\ref{eq:dp}).
\end{Theorem}
\begin{proof}
In analogy with equations (\ref{eq:v0}-\ref{duh1}) we have from lemma \ref{lemma:sol}
\begin{equation}
\label{eq:v1-nl} v_1(t,r) = \frac{h(t-r)-h(t+r)}{r}
\end{equation}
and from \refa{Fu-perteq-p}
\begin{equation}
\label{eq:v2(1)-nl} v_p(t,r) = \frac {1} {8 r} \int\limits_{|t-r|}^{t+r} d\xi \, \int
\limits_{-\xi}^{t-r} d \eta \, (\xi-\eta) f_0 (v_1(\eta, \xi))^p.
\end{equation}
As before, interchanging the order of integration we get for $t>r+R$
\begin{equation}
\label{eq:v1(2)-nl} v_p(t,r) = \frac {2^{p-3} f_0} {r} \int \limits_{-R}^{+R} d \eta \,
(h(\eta))^p \int\limits_{t-r}^{t+r} d\xi \, (\xi-\eta)^{-p+1}.
\end{equation}
Using lemma \ref{lemma:integral} we get for $t\gg r+R$ and $p \geq 3$
\begin{equation}
\label{eq:v2(3)-nl} v_p(t,r) = \frac{d_p}{t^{p-1}} + \mathcal{O} \left( \frac{r+R}{t^{p}}
\right),
\end{equation}
where
\begin{equation}
\label{eq:dp} d_p = 2^{p-2} b_0 \int \limits_{-R}^{+R} d \eta \, (h(\eta))^p.
\end{equation}
\end{proof}

Now, we will show that $v_p$ dominates the perturbation series for large times and small $\eps$
and has the same decay rate as the full solution $u$ of the nonlinear wave equation (see Table~2
and Fig.~2 for the numerical verification).

\begin{Theorem} \label{Th:Fu-u}
Under the assumptions of theorem \ref{Th:Fu-vp}, for small $\eps$ and $t\gg r+R$, we have
\begin{equation}
  u(t,r)\cong \eps^p v_p(t,r) [1+\Ocal{\eps}]\,,
\end{equation}
hence
\begin{equation}
  u(t,r) \cong D t^{-p+1}, \qquad D=d_p \eps^p + \Ocal{\eps^{p+1}}\,.
\end{equation}
\end{Theorem}

\begin{proof}
We need to show that $\eps I_0(f,g)$ and $\eps^{n+1} L_0(F_n(v_1,...,v_n))$ for $n\geq p$ are
small relative to $\eps^p d_p(x) t^{-(p-1)}$. As before, for $v_1=I_0(f,g)\in\Linftx{\infty}$
Huygens' principle and compact support of the initial data imply that $v_1(t,r)=0$ for
sufficiently large $t$ (and fixed $r$).
From the convergence proof for the perturbation series we know that there exist constants
$M,\rho>0$ such that $\|v_{n}\|_\Linftx{p-1}\leq M \rho^n$ for all $n\geq 1$. Hence, for
sufficiently small $\eps<1/\rho$ we can estimate the remainder of the perturbation series
\begin{equation}
  \left\|\sum_{m=p+1}^\infty \eps^m v_m\right\|_\Linftx{p-1}
  \leq M \sum_{m=p+1}^\infty \eps^m \rho^m
  \leq \frac{M \eps^{p+1} \rho^{p+1}}{1-\eps\rho} \leq C \eps^{p+1}\,.
\end{equation}
This implies that for $t\gg r$
\begin{equation}
  \left|\sum_{m=p+1}^\infty \eps^m v_m(t,r)\right| \leq
  \frac{C \eps^{p+1}}{\norm{t+r}\norm{t-r}^{p-2}}
  = \Ocal{\frac{\eps^{p+1}}{t^{p-1}}}\,,
\end{equation}
so
\begin{equation}
  \left| u(t,r) -  \eps^p v_p(t,r)\right| \leq
  |\eps v_1(t,r)| + \left|\sum_{m=p+1}^\infty \eps^m v_m\right|
  = \Ocal{\frac{\eps^{p+1}}{t^{p-1}}}.
\end{equation}
From theorem \ref{Th:Fu-vp}, for $t\gg r+R$ we have $v_p= d_p t^{-(p-1)} + \Ocal{t^{-p}}$ , hence
\begin{equation}
  \left|u(t,r) - \frac{d_p \eps^p}{t^{p-1}}\right|
  = \Ocal{\frac{\eps^{p+1}}{t^{p-1}}} + \Ocal{\frac{\eps^p}{t^p}},
\end{equation}
which finally gives
\begin{equation}
  u(t,r) \cong D t^{-p+1}, \qquad D=d_p \eps^p + \Ocal{\eps^{p+1}}\,.
\end{equation}
\end{proof}

\begin{table}[h]
\centering
\begin{tabular}{|c c||c|c|c|c|}
\hline
 & &\multicolumn{2}{c|}{$\varepsilon=1$}&\multicolumn{2}{c|}{$\varepsilon=10^{-1}$} \\
\cline{3-6}
 & & Theory & Numerics & Theory & Numerics \\
\hline \hline
{$p=3$} & Exponent &  2.0 & 2.0009 & 2.0 & 2.0008 \\
       & Amplitude & $0.1421$ & $0.1265$ & $0.1421 \times 10^{-4}$ & $0.1427 \times 10^{-4}$ \\
\hline
{$p=4$} & Exponent & 3.0 & 3.0013 & 3.0 & 3.0012 \\
       & Amplitude & $9.0873 \times 10^{-2}$ & $8.4433 \times 10^{-2}$ & $9.0873 \times 10^{-6}$ & $9.1631 \times 10^{-6}$ \\
\hline
{$p=5$} & Exponent & 4.0 & 4.0015 & 4.0 & 4.0015 \\
        & Amplitude & $5.9925 \times 10^{-2}$ & $6.1192 \times 10^{-2}$ & $5.9925 \times 10^{-7}$ & $6.0597 \times 10^{-7}$ \\
\hline
\end{tabular}
\caption{Nonlinear case: $F(u) = u^p$, without a potential term. The initial data are the same as
for Table~1 and Fig.~1. The number at the  Theory-Amplitude entry gives the value
 of $\epsilon^p d_p$, with $d_p$ defined in (\ref{eq:dp}) (with $b_0=1$).}
\label{tab:nl}
\end{table}

\begin{figure}[h]
\centering
\includegraphics[width=0.7\textwidth]{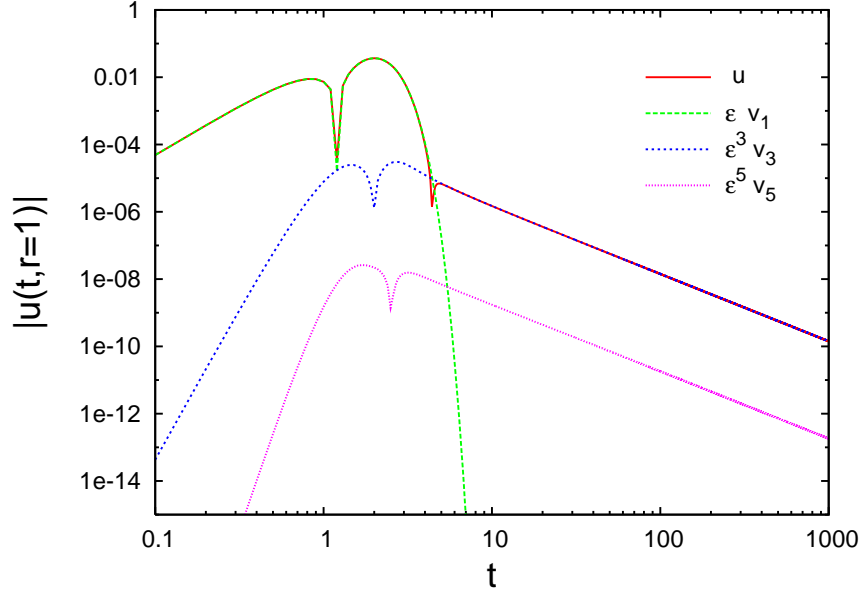}
\caption{\small{We plot (on log-log scale) the numerical solution $u(t,r=1)$ of equation
(\ref{Fu:wave-eq}) with $F(u)=u^3$. The initial data are the same as in Fig.~1  and $\varepsilon=0.1$.
The first three terms in the perturbation expansion (\ref{eq:Fu-Taylor-series}) are superimposed.
In agreement with theorem~5, the tail is perfectly approximated by $\varepsilon^3 v_3$ (cf.
Table~2).}}
\end{figure}

\section{Nonlinear case with a potential term}

Finally, let us consider the full nonlinear wave equation (\ref{wave-eq}) with a potential with
initial data $(f,g)$ supported on the interval $r\in[0,R]$ and satisfying \refa{fg-bound} with
$f_0, f_1, g_0 < \eps$. The nonlinear term $F(u)$ is the same as in the previous section.

\subsection{Perturbation series}

Defining the perturbation expansion for the nonlinear wave equation with a potential
\refa{wave-eq}
\begin{equation} \label{Fu-V:pert-ser}
   u = \sum_{n=1}^\infty \eps^n v_n\,,
\end{equation}
we encounter the problem of two scales which are given by parameters $\la$ (measuring the
strength of the potential) and $\eps$ (measuring the strength of the initial data). Since these
parameters play only an auxiliary role in generating the perturbation scheme, we make a
convenient choice and assign to $\la$ a scale of some power of $\eps$, say $\la=\w{\la} \eps^a$
with $a\in\N_+$.

Then, the power series (\ref{Fu-V:pert-ser}) inserted into the wave equation \refa{wave-eq} gives
\begin{align}
  v_{-n} &:= 0,\qquad n\geq 0 \label{V-Fu-2perteq-}\,, \\
  v_1 &:=I_0(f,g) \label{V-Fu-2perteq1}\,, \\
   v_{n+1} &:= -\w{\la} L_0(V v_{n+1-a}) + L_0(F_n(v_1,...,v_n)) \label{V-Fu-2perteq},\qquad n\geq 1.
\end{align}
In the following we choose $a:=p-1$ because then the lowest-order nontrivial term, $v_p$ (all
lower-order terms with $1<n<p$ vanish), contains contributions both from $V$ and $F$ and gives a
good approximation to $u$, as will be shown below.

In this case, from part I, we also have a convergence result
\begin{Theorem} \label{Th:V-Fu-pert2}
With $f,g, V$ and $F(u)$ as above for any $k>2$, $p\geq 3$, $\la<C_{q,k}^{-1}$ and sufficiently
small $\eps$ the series defined in \refa{Fu-V:pert-ser}-\refa{V-Fu-2perteq} converges (in norm)
in $\Linftx{q}$ for $q=\min(p-1,k)$ to the solution of the equation \refa{wave-eq} with initial
data \refa{Fu-pert-initdata}.
\end{Theorem}

\subsection{Optimal decay estimate}

For $a=p-1$ the system \refa{V-Fu-2perteq-}-\refa{V-Fu-2perteq} takes the form
\begin{align}
  v_{-n} &:= 0,\qquad n\geq 0\\
  v_1&=I_0(f,g) \label{Fu-V:pert-v1}\\
  v_2&=v_3=...=v_{p-1}=0\\
  v_{p} &= -\w{\la} L_0(V v_1) +  L_0(F_{p-1}(v_1,...,v_{p-1})) = -\w{\la} L_0(V v_1) +b_0 L_0((v_1)^p) \label{Fu-V:pert-vp}\\
  v_{n+1} &= -\w{\la} L_0(V v_{n-p+2}) +L_0(F_n(v_1,...,v_n)),\qquad n\geq p.
\end{align}
\begin{Theorem} \label{Th:FuV-vp}
Under the above assumptions, for $t\gg r+R$, we have
\begin{equation}\label{th8}
  v_p(t,r) \cong d_p\, t^{-q},\qquad q:=\min(k,p-1),
\end{equation}
where the constant $e_p$ is defined in (\ref{eq:ep}).
\end{Theorem}
\begin{proof}
$v_p$ defined in equation \refa{Fu-V:pert-vp} is a sum of two contributions, from the potential
and from the nonlinear term,
\begin{equation}
  v_p(t,r) \equiv v_p^{pot}(t,r) + v_p^{non}(t,r),
\end{equation}
where
\begin{equation}
  v_p^{pot}(t,r):=-\w{\la} L_0(V v_1),\qquad
  v_p^{non}(t,r):=b_0 L_0(v_1^p).
\end{equation}
From theorem \ref{Th:lin-v1} we have
\begin{equation}
  v_p^{pot}(t,r) = \frac{c_p}{t^k} + o\left(\frac{1}{t^k}\right)
\end{equation}
with
\begin{equation}
c_p = - 2^{k-1} \w{\la} V_0 \int \limits_{-R}^{+R} h(\eta)\,d\eta\,,
\end{equation}
and from theorem \ref{Th:Fu-vp} we have
\begin{equation}
v_p^{non}(t,r) = \frac{d_p}{t^{p-1}} + \mathcal{O} \left( \frac{r+R}{t^{p}} \right)
\end{equation}
with
\begin{equation}
d_p = 2^{p-2} b_0 \int \limits_{-R}^{+R} d \eta \, (h(\eta))^p.
\end{equation}
Depending on whether $k<p-1$ or $k>p-1$, the linear $v_p^{pot}(t,r)$ or the nonlinear
$v_p^{non}(t,r)$ contribution to the tail is dominant, respectively. In the special case $k=p-1$
we have
\begin{equation}
  v_p(t,r)=\frac{c_p+d_p}{t^{p-1}} + o\left(\frac{1}{t^{p-1}}\right).
\end{equation}
Thus, the constant in (\ref{th8}) is given by
\begin{equation}\label{eq:ep}
  e_p = \left\{\begin{array}{ll}
    c_p & \text{if }k<p-1\, \\ c_p+d_p & \text{if }k=p-1, \\ d_p & \text{if }k>p-1.
  \end{array} \right.
\end{equation}
\end{proof}

Now, we will show that $v_p$ dominates the perturbation series for large times and small $\eps$
and has the same decay rate as the full solution $u$ of the nonlinear wave equation with the
potential (see Fig.~3 for the numerical verification).

\begin{figure}[h]
\centering
\includegraphics[width=0.7\textwidth]{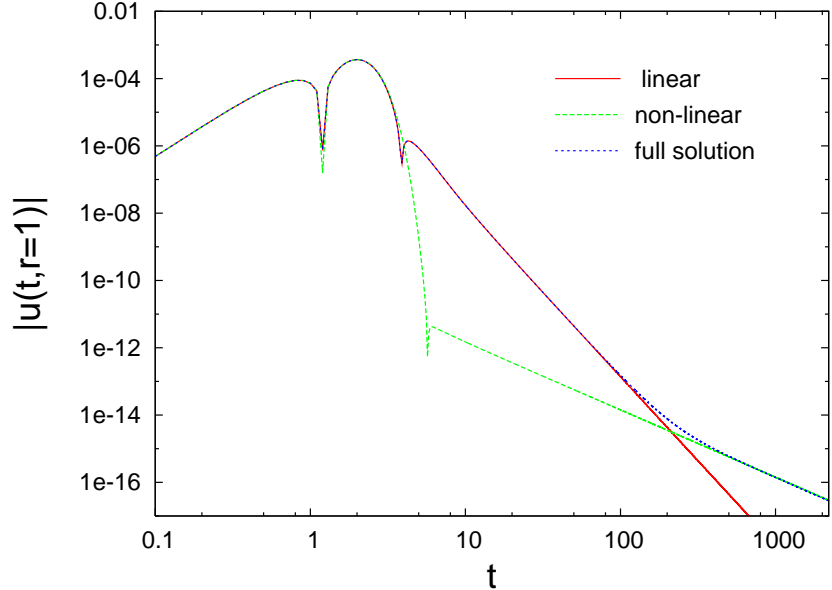}
\caption{\small{We plot (on log-log scale) the numerical solution $u(t,r=1)$ of equation
(\ref{wave-eq}) with $F(u)=u^3$. The potential $\la V $and the initial data $(\eps f,\eps g)$ are
the same as in Fig.~1 with $\la=0.64$ and $\eps=0.001$. Superimposed are solutions with the
nonlinearity or the potential switched off. The crossover from the linear tail$~\sim t^{-5}$ (for
intermediate times) to the final nonlinear tail$~\sim t^{-2}$ is clearly seen. }}
\end{figure}

\begin{Theorem} \label{Th:FuV-u}
Under the assumptions of theorem \ref{Th:FuV-vp}, for small $\eps$ and $t\gg r+R$, we have
\begin{equation}
  u(t,r)\cong \eps^p v_p(t,r) [1+\Ocal{\eps}]\,,
\end{equation}
hence
\begin{equation}
  u(t,r) \cong E t^{-q}, \qquad q:=\min(k,p-1),\qquad E=e_p \eps^p + \Ocal{\eps^{p+1}}\,.
\end{equation}

\end{Theorem}
\begin{proof}
We can repeat the reasoning from the proof of theorem \ref{Th:Fu-u} where we used the fact that
the perturbation series $\sum_{n=1} \eps^n v_n$ is convergent. Here, theorem \ref{Th:V-Fu-pert2}
guarantees convergence in $\Linftx{q}$ with $q:=\min(k,p-1)$. Analogously, we obtain for $t\gg r$
\begin{equation}
  \left|\sum_{m=p+1}^\infty \eps^m v_m(t,r)\right| \leq
  \frac{C \eps^{p+1}}{\norm{t+r}\norm{t-r}^{q-1}}
  = \Ocal{\frac{\eps^{p+1}}{t^{q}}}\,,
\end{equation}
so (again, $v_1(t,r)$ vanishes for $t\gg r$ by  Huygens' principle)
\begin{equation}
  \left| u(t,r) -  \eps^p v_p(t,r)\right| \leq
  |\eps v_1(t,r)| + \left|\sum_{m=p+1}^\infty \eps^m v_m\right|
  = \Ocal{\frac{\eps^{p+1}}{t^{q}}}.
\end{equation}
From theorem \ref{Th:FuV-vp} we have $v_p= e_p t^{-q} + o(t^{-q})$ for $t\gg r+R$, so we get
\begin{equation}
  \left|u(t,r) - \frac{e_p \eps^p}{t^{q}}\right|
  = \Ocal{\frac{\eps^{p+1}}{t^{q}}} + o\left(\frac{\eps^p}{t^q}\right),
\end{equation}
which gives
\begin{equation}
  u(t,r) \cong E t^{-q}, \qquad E=e_p \eps^p + \Ocal{\eps^{p+1}}.
\end{equation}
\end{proof}

\begin{acknowledgments}
PB acknowledges hospitality of the Albert-Einstein-Institute in Potsdam, where this work has been
started. NS acknowledges hospitality of the Institute of Physics, Jagellonian University in
Cracow, where this work has been finished. This research was supported in part by the MNII grant
1 P03B 012 29.
\end{acknowledgments}

\appendix*
\section{Lemmas}

\begin{Lemma} \label{lemma:sol}
The solution of the free wave equation
\begin{equation} \label{lemma:sol:waveeq}
  \Box u = 0
\end{equation}
with spherically symmetric initial data $u(0,r)=f(r)$, $\d_t u(0,r)=g(r)$ has the form
\begin{equation}
  u(t,r)=\frac{h(t-r)-h(t+r)}{r},
\end{equation}
where
\begin{equation} \label{lemma:sol:h}
  h(r) = -\frac{r}{2} f(r) + \frac{1}{2}\int_{r}^{\infty} r' g(r') dr',
\end{equation}
which is defined for all $r\in \R$ by the extension $f(-r):=f(r)$, $g(-r):=g(r)$. When $f$ and
$g$ have compact support then $h$ has also compact support on $\R$.
\end{Lemma}
\begin{proof}
In spherical symmetry the wave equation \refa{lemma:sol:waveeq} can be written as
\begin{equation}
  \d_\xi \d_\eta (ru) = 0
\end{equation}
where $\xi = t + r$ and $\eta = t - r$. Its most general solution has the form
\begin{equation}
  ru(t,r) = h_-(\eta) + h_+(\xi) = h_-(t-r) + h_+(t+r).
\end{equation}
We require that $u(t,r)$ be finite at $r=0$ what implies
\begin{equation}
  0 = h_-(t) + h_+(t) \qquad \Rightarrow \qquad h(t):=h_{-}(t)=-h_{+}(t).
\end{equation}
From the initial conditions we get
\begin{align}
\label{lemma:sol:f}
  f(r) &= u(0,r) = \frac{h(-r)-h(r)}{r} \\
\label{lemma:sol:g}
  g(r) &= \d_t u(0,r) = \frac{h'(-r)-h'(r)}{r}.
\end{align}
We can write
\begin{equation}
h(r) = \frac {1} {2} [h(r)+h(-r)] + \frac {1} {2} [h(r)-h(-r)] \equiv \frac {1} {2} S(r) + \frac {1} {2} A(r),
\end{equation}
where we have introduced a symmetric function $S(r):=h(r)+h(-r)$ and an anti-symmetric function $A(r):=h(r)-h(-r)$. The solutions for $A(r)$ and $S'(r)=h'(r)-h'(-r)$ can be immediately read off from the initial conditions (\ref{lemma:sol:f}, \ref{lemma:sol:g}):
\begin{align}
\label{lemma:sol:A}
  A(r) &= -r f(r) \\
\label{lemma:sol:Sbis}
  S'(r) &= -r g(r).
\end{align}
We see that the extension of $f$ and $g$ on all $r\in\R$ defined by $f(-r):=f(r)$, $g(-r):=g(r)$ is the consistency condition for eqs. (\ref{lemma:sol:A}, \ref{lemma:sol:Sbis}). Integrating (\ref{lemma:sol:Sbis}) we get
\begin{equation}
S(r) - S(0) = - \int_0^{r} r' g(r') dr'.
\end{equation}
We use the freedom of choosing the integration constant and set
\begin{equation}
  S(0):=\int_{0}^{\infty} r' g(r') dr',
\end{equation}
what gives \refa{lemma:sol:h}. With this choice we obtain $h(r)$ compactly supported on $\R$ if
$f(r)$ and $g(r)$ are compactly supported. To see this, assume $f(x)=g(x)=0$ for $|x|>R$ and
consider $r>R$. The function $h(r)$ is obviously zero from \refa{lemma:sol:h}. For negative
arguments
\begin{equation}
  h(-r)= \frac{r}{2} \underbrace{f(-r)}_{=0} + \frac{1}{2}\int_{-r}^{\infty} r' g(r') dr'
  = \frac{1}{2}\int_{-R}^{R} r' g(r') dr' = 0\,,
\end{equation}
because the integrand $r' g(r')$ is an odd function. Thus, $\text{supp}\; h \subset [-R,+R]$.
\end{proof}

\begin{Lemma} \label{lemma:integral}
Let $\alpha>1$. Then
\begin{equation}
  \int \limits_{-R}^{+R} h(\eta)\,d\eta \int \limits_{t-r}^{t+r} \frac{d\xi}{(\xi - \eta)^{\alpha}}
  = \frac {2 r}{t^{\alpha}} \int_{-R}^{+R} h(\eta)\,d\eta
  + \Ocal{\frac{r(r+R)}{t^{\alpha+1}}}
\end{equation}
for $t> 2\alpha(r + R)$ and all $r\geq 0$.
\end{Lemma}

\begin{proof}
Consider first the inner integral for $\eta\in[-R,R]$
\begin{equation}
\begin{split}
  I(t,r,\eta):=&\int_{t-r}^{t+r} \frac{d\xi}{(\xi - \eta)^{\alpha}}
  = \int_{-r}^{+r} \frac{dy}{(t-\eta+y)^{\alpha}}
  = \frac{1}{t^\alpha} \int_{-r}^{+r} \frac{dy}{\left(1+\frac{y-\eta}{t}\right)^{\alpha}} \\
  =&\frac {2 r}{t^{\alpha}}  + \frac{1}{t^{\alpha}}
    \int_{-r}^{+r} \left[\frac{1}{\left(1+\frac{y-\eta}{t}\right)^{\alpha}} - 1\right]\,dy
  \equiv \frac {2 r}{t^{\alpha}}  + \frac{1}{t^{\alpha}} \delta(t,r,\eta).
\end{split}
\end{equation}
We have
\begin{equation}
  |\delta(t,r,\eta)| \leq
  \int_{-r}^{+r} \left|\frac{1}{\left(1+\frac{y-\eta}{t}\right)^{\alpha}} - 1\right|\,dy
\end{equation}
and the integrand $J$ can be estimated by
\begin{equation} \label{lemma:integral:J}
  J:=\left|\frac{1}{\left(1+\frac{y-\eta}{t}\right)^{\alpha}} - 1\right| \leq
  \frac{1}{\left(1-\frac{r+R}{t}\right)^{\alpha}} - 1,
\end{equation}
what can be shown as follows. Having in mind that $-(r+R)\leq y+\eta\leq r+R$ we find for
$y-\eta<0$
\begin{equation}
  J = \frac{1}{\left(1+\frac{y-\eta}{t}\right)^{\alpha}} - 1 \leq
  \frac{1}{\left(1-\frac{r+R}{t}\right)^{\alpha}} - 1 \equiv J_1
\end{equation}
and for $y-\eta\geq 0$
\begin{equation}
  J =  1- \frac{1}{\left(1+\frac{y-\eta}{t}\right)^{\alpha}} \leq
  1-\frac{1}{\left(1+\frac{r+R}{t}\right)^{\alpha}} \equiv J_2.
\end{equation}
By simple algebra one can easily show that $J_2\leq J_1$ what gives \refa{lemma:integral:J}.
Further, by a version of Bernoulli's inequality,
\begin{equation}
  J\leq \frac{1}{1-\alpha \frac{r+R}{t}} - 1   \frac{\alpha\frac{r+R}{t}}{1-\alpha\frac{r+R}{t}}.
\end{equation}
Then
\begin{equation}
  |\delta(t,r,\eta)| \leq \frac{\alpha\frac{r+R}{t}}{1-\alpha\frac{r+R}{t}} 2r
  \leq 4 \alpha\frac{r(r+R)}{t}
\end{equation}
for $t\geq 2\alpha(r+R)$. Finally,
\begin{equation}
\begin{split}
  \int_{-R}^{+R} h(\eta)\,d\eta \int_{t-r}^{t+r} \frac{d\xi}{(\xi - \eta)^{\alpha}}
  &= \int_{-R}^{+R} h(\eta)\,d\eta \left[
    \frac {2 r}{t^{\alpha}}  + \frac{\delta(t,r,\eta)}{t^{\alpha}} \right]\\
  &= \frac {2 r}{t^{\alpha}} \int_{-R}^{+R} h(\eta)\,d\eta
  + \Ocal{\frac{r(r+R)}{t^{\alpha+1}}}\,.
\end{split}
\end{equation}
\end{proof}



\end{document}